\RequirePackage{ifpdf}
\ifpdf 
\documentclass[pdftex]{sigma}
\else
\documentclass{sigma}
\fi

\newcommand{\comp}{\leftrightarrow}

\newcommand{\ord}{\operatorname{ord}}
\def\comp{\leftrightarrow}
\newcommand{\notcomp}{\not\kern-0.2ex\comp}

\usepackage[cp1250]{inputenc}
\begin{document}

\allowdisplaybreaks

\renewcommand{\thefootnote}{$\star$}

\renewcommand{\PaperNumber}{00?}

\FirstPageHeading

\ShortArticleName{Sharply Orthocomplete Ef\/fect Algebras}

\ArticleName{Sharply Orthocomplete  Ef\/fect Algebras\footnote{This paper is a
contribution to the Proceedings of the 6-th Microconference
``Analytic and Algebraic Me\-thods~VI''. }}

\Author{Martin KALINA, Jan PASEKA and Zdenka RIEČANOVÁ}

\AuthorNameForHeading{M. Kalina, J. Paseka and Z. Riečanová}

\Address{Department of Mathematics,
Faculty of Civil Engineering, Slovak University \\
of Technology,
Radlinsk\' eho 11, SK-813~68~Bratislava, Slovakia}
\Email{\href{kalina@math.sk}{kalina@math.sk}} 

\Address{Department of Mathematics
 and Statistics,
Faculty of Science,
Masaryk University,\\
Kotl\'a\v{r}sk\'a~2,
CZ-611~37~Brno, Czech Republic}
\Email{\href{mailto:paseka@math.muni.cz}{paseka@math.muni.cz}} 

\Address{Department of Mathematics, Faculty of Electrical Engineering 
and Information\\ Technology,
Slovak University of Technology, Ilkovi\v{c}ova~3, SK-812~19~Bratislava,\\
Slovak Republic}
\Email{\href{zdena.riecanova@gmail.com}{zdena.riecanova@gmail.com}}


\Abstract{Special types of ef\/fect algebras $E$ called 
sharply dominating and S-dominating were introduced by S. Gudder 
in \cite{gudder1,gudder2}. We prove statements about 
connections between sharp orthocompleteness, sharp dominancy 
and completeness of $E$. Namely we prove that in every sharply orthocomplete 
S-dominating effect algebra $E$ the set of sharp elements and 
the center of $E$ are complete lattices bifull in $E$. If an 
Archimedean atomic  lattice effect algebra 
$E$ is sharply orthocomplete then it is complete.}


\Classification{06C15;  03G12; 81P10}

\section{Introduction}\label{intro}

An algebraic structure called an effect algebra has been introduced by D.J. 
Foulis and M.K. Bennett (1994). The advantage of an effect algebra is that 
effect algebras provide a mechanism for studying quantum effects, or more general, 
in non-classical probability theory their elements represent events that may 
be unsharp or pairwise non-compatibble. Lattice effect algebras are in 
some sence a nearest common generalization of orthomodular lattices \cite{kalmbach} 
that may include non-compatible pairs of elements, and $MV$-algebras \cite{C.C.Ch} 
that may include unsharp elements. More precisely a lattice effect algebra $E$  
is an orthomodular lattice iff every element of $E$ is sharp 
(i.e., $x$ and "non $x$" are disjoint) and it is 
an $MV$-effect algebra iff every pair of elements of E is compatible. Moreover, 
in every lattice effect algebra $E$ the set of sharp elements is an orthomodular 
lattice (\cite{ZR57}), 
and E is a union of its blocks (i.e., maximal 
subsets of pairwise compatible elements  that are $MV$-effect algebras 
(see \cite{ZR56})). Thus a lattice effect algebra $E$ is a Boolean algebra 
iff every pair of elements are compatible and every element of $E$ is sharp. 

However, non-lattice ordered effect algebra $E$ is so general 
that its set $S(E)$ of sharp 
elements may  form neither an orthomodular lattice nor any regular 
algebraic structure. S. Gudder (see \cite{gudder1,gudder2}) introduced 
special types of effect algebras $E$ called sharply dominating, whose 
set $S(E)$ of sharp elements forms an 
orthoalgebra and also so called S-dominating ,whose set $S(E)$ of sharp elements 
forms an orthomodular lattice. In \cite{gudder1}, S.~Gudder showed that a standard 
Hilbert space effect algebra ${\mathcal E}(H)$ of bounded operators on a Hilbert space $H$ 
between zero and identity operators (with partially defined usual operation + ) are 
S-dominating. Hence S-dominating effect algebras may be useful abstract models 
for sets of quantum effects in physical systems.

We study these two special kinds of effect algebras. We show properties of 
some remarkable sub-effect algebras of such effect algebras $E$ satisfying the 
condition that $E$ is sharply orthocomplete. Namely properties of their blocks, 
sets of sharp elements and their centers. It is worth to note that in 
\cite{kopr} it was proved that there are even Archimedean atomic $MV$-effect 
algebras which are not sharply dominating hence they are not S-dominating.

\section{Basic def\/initions and some known facts}\label{basic}

\begin{definition}[\cite{FoBe}]\label{def:EA}
A partial algebra $(E;\oplus,0,1)$ is called an {\em effect algebra} if
$0$, $1$ are two distinct elements and $\oplus$ is a partially
def\/ined binary operation on $E$ which satisfy the following
conditions for any $x,y,z\in E$:
\begin{description}\itemsep=0pt
\item[$(Ei)$\phantom{iii}] $x\oplus y=y\oplus x$ if $x\oplus y$ is def\/ined,
\item[$(Eii)$\phantom{ii}] $(x\oplus y)\oplus z=x\oplus(y\oplus z)$  if one
side is def\/ined,
\item[$(Eiii)$] for every $x\in E$ there exists a unique $y\in
E$ such that $x\oplus y=1$ (we put $x'=y$),
\item[$(Eiv)$\phantom{i}] if $1\oplus x$ is def\/ined then $x=0$.
\end{description}

\end{definition}

We often denote the ef\/fect algebra $(E;\oplus,0,1)$ brief\/ly by
$E$. On every ef\/fect algebra $E$  the partial order
$\le$  and a partial binary operation $\ominus$ can be
introduced as follows:
\[
x\le y  \mbox{ and }  y\ominus x=z  \mbox{ if\/f } x\oplus z
\mbox{ is def\/ined and } x\oplus z=y .
\]

If $E$ with the def\/ined partial order is a lattice (a complete
lattice) then $(E;\oplus,0,1)$ is called a {\em lattice effect
algebra} ({\em a complete lattice effect algebra}).

\begin{definition}\label{subef}
Let $E$ be an  ef\/fect algebra.
Then $Q\subseteq E$ is called a {\em sub-effect algebra} of  $E$ if
\begin{enumerate}\itemsep=0pt
\item[$(i)$] $1\in Q$,
\item[$(ii)$] if out of elements $x,y,z\in E$ with $x\oplus y=z$
two are in $Q$, then $x,y,z\in Q$.
\end{enumerate}
If $E$ is a lattice ef\/fect algebra and $Q$ is a sub-lattice and a sub-ef\/fect
algebra of $E$ then $Q$ is called a~{\em sub-lattice effect algebra} of $E$.
\end{definition}

Note that a sub-ef\/fect algebra $Q$
(sub-lattice ef\/fect algebra $Q$) of an  ef\/fect algebra $E$
(of a~lattice ef\/fect algebra $E$) with inherited operation
$\oplus$ is an  ef\/fect algebra (lattice ef\/fect algebra)
in its own right.


For an element $x$ of an ef\/fect algebra $E$ we write
$\ord(x)=\infty$ if $nx=x\oplus x\oplus\dots\oplus x$ ($n$-times)
exists for every positive integer $n$ and we write $\ord(x)=n_x$
if $n_x$ is the greatest positive integer such that $n_xx$
exists in $E$.  An ef\/fect algebra $E$ is {\em Archimedean} if
$\ord(x)<\infty$ for all $x\in E$.

A minimal nonzero element of an ef\/fect algebra  $E$
is called an {\em atom}  and $E$ is
called {\em atomic} if under every nonzero element of
$E$ there is an atom.

For a poset $P$ and its subposet $Q\subseteq P$ we denote,
for all $X\subseteq Q$, by $\bigvee_{Q} X$ the join of
the subset $X$ in the poset $Q$ whenever it exists.

We say that a f\/inite system $F=(x_k)_{k=1}^n$ of not necessarily
dif\/ferent elements of an ef\/fect algebra $(E;\oplus,0,1)$ is
{\em orthogonal} if $x_1\oplus x_2\oplus \cdots\oplus
x_n$ (written $\bigoplus\limits_{k=1}^n x_k$ or $\bigoplus F$) exists
in $E$. Here we def\/ine $x_1\oplus x_2\oplus \cdots\oplus x_n=
(x_1\oplus x_2\oplus \cdots\oplus x_{n-1})\oplus x_n$ supposing
that $\bigoplus\limits_{k=1}^{n-1}x_k$ is def\/ined and
$\bigoplus\limits_{k=1}^{n-1}x_k\le x'_n$. We also def\/ine
$\bigoplus \varnothing=0$.
An arbitrary system
$G=(x_{\kappa})_{\kappa\in H}$ of not necessarily dif\/ferent
elements of $E$ is called {orthogonal} if $\bigoplus K$
exists for every f\/inite $K\subseteq G$. We say that for an {orthogonal}
system $G=(x_{\kappa})_{\kappa\in H}$ the
element $\bigoplus G$ (more precisely $\bigoplus_{E} G$) exists if\/f
$\bigvee\{\bigoplus K
\mid
K\subseteq G$ is f\/inite$\}$ exists in $E$ and then we put
$\bigoplus G=\bigvee\{\bigoplus K\mid K\subseteq G$ is
f\/inite$\}$. (Here we write $G_1\subseteq G$ if\/f there is
$H_1\subseteq H$ such that $G_1=(x_{\kappa})_{\kappa\in
H_1}$).

We call an ef\/fect algebra $E$ {\em orthocomplete} \cite{JePu} if every 
orthogonal system $G=(x_{\kappa})_{\kappa\in H}$ of elements of $E$ 
has the sum $\bigoplus G$. It is known that every orthocomplete 
Archimedean lattice ef\/fect algebra $E$  is a complete lattice 
(see \cite[Theorem 2.6]{ZR60}).

Recall that elements $x, y$ of a lattice ef\/fect algebra
$E$ are called {\em compatible} (written
$x\comp y$) if\/f
$x\vee y=x\oplus(y\ominus(x\wedge y))$ (see \cite{kopkabool}).
$P\subseteq E$ is a {\it set of pairwise compatible elements}
if $x\leftrightarrow y$ for all $x,y\in P$.  $M\subseteq E$ is
called a {\it block} of $E$ if\/f $M$ is a maximal subset of
pairwise compatible elements. Every block of a lattice ef\/fect
algebra $E$ is a sub-ef\/fect algebra and a~sub-lattice of $E$ and
$E$ is a union of its blocks (see \cite{ZR56}).
Lattice ef\/fect algebra with a unique block is called an
{\it $MV$-effect algebra}. Every block of a lattice ef\/fect
algebra is an $MV$-ef\/fect algebra in its own right.

An element $w$ of an ef\/fect algebra $E$ is called
{\em sharp} (see (\cite{gudder1,gudder2})) if $w\wedge w'=0$.

\begin{definition} (\cite{gudder1,gudder2})
An  effect algebra $E$ is called 
{\bfseries{\em sharply dominating}} if 
for every $x\in E$\ {there exists} 
$\widehat{x}\in S(E)$ such that 
\begin{center}$\widehat{x}=\bigwedge_{E}\{w\in S(E) \mid x\leq w\}=
\bigwedge_{S(E)}\{w\in S(E) \mid x\leq w\}.$
\end{center}

Note that clearly $E$ is sharply dominating iff 
for every $x\in E$\ {there exists} 
$\widetilde{x}\in S(E)$ such that 
\begin{center}$\widetilde{x}=\bigvee_{E}\{w\in S(E) \mid x\geq w\}=
\bigvee_{S(E)}\{w\in S(E) \mid x\geq w\}.$ 
\end{center} 

A sharply dominating effect algebra $E$ 
is called {\bfseries{\em S-dominating}} (\cite{gudder2})
if $x \wedge w$ exists for every $x\in E$, 
$w\in S(E)$.

\end{definition}

The well known fact is that in every S-dominating ef\/fect algebra $E$
the subset $S(E)=\{w\in E\mid w\wedge w'=0\}$  of sharp elements 
of $E$ is a sub-ef\/fect 
algebra of $E$ being an orthomodular lattice (see \cite[Theorem 2.6]{gudder2}). 
Moreover if for $D\subseteq S(E)$ the element $\bigvee_{E}D$ exists then
 $\bigvee_{E}D\in S(E)$  hence $\bigvee_{S(E)}D=\bigvee_{E}D$.
We say that $S(E)$  is a full sublattice of $E$ (see \cite{ZR57}).

Let $G$ be a sub-effect algebra of an effect algebra $E$. We say 
that $G$ is {\bfseries{\em bifull in}} $E$, if, for any 
$D\subseteq G$ the element
$\bigvee_{G} D$ exists if\/f the element $\bigvee_{E} D$ exists and
they are equal. Clearly, any bifull sub-effect algebra of $E$ is 
full but not conversely (see \cite{kalina}). 


The notion of a central element of an effect algebra
$E$ was introduced by Greechie-Foulis-Pulmannov\'a
\cite{GrFoPu}. An element $c\in E$ is called
{\em central} (see \cite{ZR51}) iff 
for every $x\in E$ there exist $x\wedge c$ and
$x\wedge c'$  and $x=(x\wedge c)\vee(x\wedge c')$. 
The {\em center} $C(E)$ of $E$ is the set of all central elements of $E$. 
Moreover, $C(E)$ is a Boolean algebra, see~\cite{GrFoPu}.  
If $E$ is a lattice ef\/fect algebra then $z\in E$ is
central if\/f $z\wedge z'=0$ and $z\comp x$ for all $x\in E$, see \cite{ZR52}.
Thus in a lattice ef\/fect algebra $E$, $C(E)=B(E)\cap
S(E)$, where $B(E)=\bigcap\{M\subseteq E\mid M\text{ is a block
of }E\}$ is called {\em compatibility center} of $E$.

An effect algebra $E$ is called 
{\bfseries{\em centrally dominating}} (see also \cite{foulis2009} for the notion 
{\em central cover}) if 
for every $x\in E$\ {there exists} 
${c}_x\in C(E)$ such that 
\begin{center}${c}_x=\bigwedge_{E}\{c\in C(E) \mid x\leq c\}=
\bigwedge_{C(E)}\{c\in C(E) \mid x\leq c\}.$
\end{center}

An element $a$ of a lattice $L$ is called 
{\em compact} iff, for any $D\subseteq L$, $a\leq \bigvee D$ implies 
$a\leq \bigvee F$ for some finite $F\subseteq D$.
A lattice $L$ is called {\em compactly generated} iff every 
element of $L$ is a join of compact elements.

\section{Sharply orthocomplete ef\/fect algebras}


In an effect algebra $E$ the set 
$S(E)=\{x\in E\mid x\wedge x'=0\}$ of sharp elements plays an 
important role. In some sense we can say that an effect algebra $E$ is a "smeared 
set $S(E)$" of its sharp elements, while unsharp effects are important in studies 
of unsharp measurements (\cite{FoBe, BLM}{}). 
S.Gudder proved (see \cite{gudder2}) that, in standard Hilbert space 
effect algebra ${\mathcal E}(H)$ of bounded operators $A$ on a Hilbert space $H$ between null 
operator and identity operator, which are endowed with usual $+$ 
defined iff $A+B$ 
is in ${\mathcal E}(H)$, the set $S({\mathcal E}(H))$ of sharp elements forms 
an orthomodular lattice of 
projections operators on $H$. Further in (see \cite[Theorem 2.2]{gudder2}) 
it was shown 
that in every sharply dominating effect algebra the set $S(E)$ is a sub-effect 
algebra of $E$. Moreover, in \cite[Theorem 2.6]{gudder1} it is proved that in every 
S-dominating effect algebra $E$ the set $S(E)$ is an orthomodular lattice. 
We are going to show that in this case $S(E)$ is bifull in $E$.

\begin{theorem} \label{shdisbifull} Let $E$ be an S-dominating  
 effect algebra.
Then $S(E)$ 
is bifull in $E$.
\end{theorem}
\begin{proof} Let $S\subseteq S(E)$.

(1) Assume that $z=\bigvee_{S(E)} S\in S(E)$ exists. 
Let us show that $z$ is the least upper bound of $S$ 
in $E$. Let $y\in E$ be an upper bound of $S$.
Then $y\wedge z$ exists and it is an upper bound of $S$ as well. 
Hence, for any $s\in S$, $s\leq y\wedge z$. This 
yields that  $s\leq  \widetilde{y\wedge z} \leq  y\wedge z$, 
for all $s\in S$, $\widetilde{y\wedge z}\in S(E)$. 
Hence $z\leq \widetilde{y\wedge z} \leq  y\wedge z \leq z$. 
Then $z=y\wedge z \leq y$ i.e., $z$
 is really the least upper bound of $S$ in $E$.

(2) Conversely, let $z=\bigvee_{E} S\in E$ exists. 
Let $y\in S(E)$ be an upper bound of $S$ in $S(E)$. 
Then $y\wedge z$ exists and it is again an upper bound of $S$. 
As $E$ is sharply dominating, there exists a greatest 
sharp element $\widetilde{y\wedge z} \leq  y\wedge z$ 
and hence $s\leq  \widetilde{y\wedge z} \leq  y\wedge z$.
This gives that  $z=\widetilde{y\wedge z}\in S(E)$. 
Thus $z=\bigvee_{S(E)} S\in S(E)$.
\end{proof}

\begin{corollary} \label{shdisbifullcv} If $E$ is a sharply dominating  
 lattice effect algebra then $S(E)$ is bifull in $E$.
\end{corollary}

\begin{definition}\label{shoc}
An effect algebra $E$ is called 
{\bfseries{\em  sharply orthocomplete }}  
({\bfseries{\em  centrally orthocomplete }} (see \cite{foulis2009}))
if for any system
$(x_{\kappa})_{\kappa\in H}$ of elements of $E$ such that there 
exists an orthogonal system  $(w_{\kappa})_{\kappa\in H}, w_{\kappa}\in S(E)$ 
with $x_{\kappa}\leq w_{\kappa}$, ${\kappa\in H}$ 
(an orthogonal system  $(c_{\kappa})_{\kappa\in H}, c_{\kappa}\in C(E)$ 
with $x_{\kappa}\leq c_{\kappa}$, ${\kappa\in H}$) there exists 
$$\bigoplus \{x_{\kappa} \mid \kappa\in H\}=%
\bigvee_{E} \{\bigoplus_{E} \{x_{\kappa} \mid \kappa\in F\} \mid %
F\subseteq H, F \ \mbox{finite}\}.$$
\end{definition}

\begin{theorem}\label{shoiscc} Let $E$ be a sharply orthocomplete 
S-dominating effect algebra. Then  
\begin{enumerate}\itemsep=0pt
\item[$(i)$] $S(E)$ is a complete orthomodular lattice  bifull  in $E$. 
\item[$(ii)$] $C(E)$ is a complete Boolean algebra  bifull  in $E$. 
\item[$(iii)$] $E$ is 
centrally dominating and centrally orthocomplete. 
\item[$(iv)$] If $C(E)$ is atomic then
$\bigvee_{E}\{p\in C(E) \mid p\ \mbox{atom of}\ C(E)\}=1$.
\end{enumerate}
\end{theorem}
\begin{proof}  $(i)$: From \cite[Theorem 2.6]{gudder2} we know that 
$S(E)$ is an orthomodular lattice and a sub-lattice effect 
algebra of $E$.

Let us show that $S(E)$ is orthocomplete.
Let $S\subseteq S(E)$, $S$ orthogonal. Then for every finite 
$F\subseteq S$ we have that 
$\bigoplus_{E} F=\bigvee_{E} F=\bigvee_{S(E)} F\in S(E)$. 
Moreover, for any $s\in S$, $s\leq s$. Since $S(E)$ is bifull in $E$ 
by Theorem \ref{shdisbifull} and $E$ is  sharply orthocomplete 
we have that $\bigoplus_{E} S=\bigvee_{E} S=\bigvee_{S(E)} S\in S(E)$ exists. Since   $S(E)$ is an Archimedean lattice effect algebra we have 
from \cite[Theorem 2.6]{ZR60} that   $S(E)$ is complete. 

\medskip

\noindent{}$(ii)$: 
As $C(E)=\{x\in E \mid y=(y\wedge x)\vee (y\wedge x')\ \mbox{for}\ 
\mbox{every}\ y\in E\}$, we obtain that 
$1=x\vee x'$ for every $x\in C(E)$ and by de Morgan Laws 
$0=x\wedge x'$ for every $x\in C(E)$. 
Hence $C(E)\subseteq S(E)$. It follows by $(i)$ that, 
for any  $Q\subseteq C(E)$, there exists 
$\bigvee_{S(E)} Q=\bigvee_{E} Q\in C(E)$ because $C(E)$ 
is full in E, hence $\bigvee_{C(E)} Q=\bigvee_{E} Q$. 
By de Morgan Laws there exists 
$\bigwedge_{E} Q=(\bigvee_{E} Q')'$, where evidently 
$Q'=\{q'\in E\mid q\in Q\}\subseteq C(E)$. Hence 
$\bigwedge_{E} Q\in C(E)$ which gives 
$\bigwedge_{C(E)} Q=\bigwedge_{E} Q$ (see also \cite{foulis2009}).

\medskip

\noindent{}$(iii)$: Let $x\in E$.  Using $(ii)$ let us put 
${c}_x=\bigwedge_{C(E)}\{c\in C(E) \mid x\leq c\}\in C(E)$. 
Since $C(E)$ is bifull  in $E$ we have that 
${c}_x=\bigwedge_{E}\{c\in C(E) \mid x\leq c\}$ (see again \cite{foulis2009}). 
Since $C(E)\subseteq S(E)$ we immediately obtain that 
$E$ is centrally orthocomplete.

\medskip

\noindent{}$(iv)$: Since $C(E)$ is an atomic Boolean 
algebra we have 
$\bigvee_{C(E)}\{p\in C(E) \mid p\ \mbox{atom of}\ C(E)\}=1$. 
As $C(E)$ is bifull  in $E$, we have that
$\bigvee_{E}\{p\in C(E) \mid p\ \mbox{atom of}\ C(E)\}=
\bigvee_{C(E)}\{p\in C(E) \mid p\ \mbox{atom of}\ C(E)\}=1$.
\end{proof}

\section{Sharply orthocomplete lattice ef\/fect algebras}

M. Kalina in \cite{kalina} has shown that even 
in an Archimedean atomic lattice effect algebra $E$ with atomic center $C(E)$ the 
join of atoms of $C(E)$ computed in $E$ need not be equal to $1$. 
Next examples and theorems show 
connections between sharp orthocompleteness, sharp dominancy and 
completeness of an effect algebra $E$ as well as bifullness 
of $S(E), C(E)$ and atomic 
blocks in a lattice effect algebra $E$.

\begin{example}
{\itshape  Example of a compactly 
generated sharply orthocomplete $MV$-effect algebra  
that is not  complete.}%

\medskip

{\rm It is enough to take the Chang $MV$-effect algebra 
$E=\{0, a, 2a, 3a, \dots, (3a)', (2a)', a', 1\}$ that is not Archimedean 
(hence non-complete), it is  compactly 
generated (every $x\in E$ is compact) and obviously 
sharply orthocomplete (the center $C(E)=S(E)$ is trivial) and hence sharply dominating.}
\end{example}

\begin{example}{\itshape  Example of a sharply dominating 
Archimedean atomic lattice $MV$-effect algebra 
$E$ with complete and bifull $S(E)$ that is not  sharply orthocomplete.}

\medskip
{\rm Let 
${E}= \prod \{\{0_n, a_n, 1_n\}\mid n=1, 2, \dots\}$ and let  
$${E}_{0}= \{(x_n)_{n=1}^{\infty}\in E \mid x_k=a_k\ 
\mbox{for at most}\ \mbox{finitely many} \
k\in \{1, 2, \dots \}\}.$$  
Then $ {E}_{0}$ is a sub-lattice effect algebra of $E$
(hence it is an $MV$-effect algebra), evidently sharply dominating  and 
it is not  sharply orthocomplete  (since it is non-complete). 

${S({E}_{0})}= \prod \{\{0_n, 1_n\}\mid n=1, 2, \dots\}$  is a 
complete Boolean algebra and 
${S({E}_{0})}={C({E}_{0})}$ is a bifull sub-lattice 
of ${E}_{0}$. }
\end{example}

\begin{lemma}\label{soaac} Let $E$ be a sharply orthocomplete  Archimedean 
atomic $MV$-effect algebra. Then  $E$ is complete.
\end{lemma}
\begin{proof} Let $A\subseteq E$ be a set of all atoms of $E$. 
Then $1=\bigvee_{E} \{n_a a | a\in A\}=%
\bigoplus_{E} \{n_a a | a\in A\}$, $n_a a\in C(E)=S(E)$ are 
atoms of $C(E)$ for all $a\in A$. By \cite[Theorem 3.1]{ZR71}
we have that  $E$ is isomorphic to a~subdirect product
of the family $\{ [0, n_a a] \mid a\in A\}$. 
The corresponding lattice effect algebra embedding $\varphi:E \to \prod\{ [0, n_a a] \mid a\in A\}$ 
is given by $\varphi(x)=(x\wedge n_a a)_{a\in A}$.

Let us check that 
$E$ is isomorphic to $\prod\{ [0, n_a a] \mid a\in A\}$. 
It is enough to check that $\varphi$ is onto. Let 
$(x_a)_{a\in A}\in \prod\{ [0, n_a a] \mid a\in A\}$. Then 
$(x_a)_{a\in A}$ is an orthogonal system and $x_a=k_a a\leq n_a a\in S(E)$ for 
all $a\in A$. Hence 
$x=\bigoplus_{E} \{x_{a} \mid a\in A\}=%
\bigvee_{E} \{k_a a  \mid a\in A\}\in E$ exists. 
Evidently, $\varphi(x)=(x\wedge n_a a)_{a\in A}=(k_a a)_{a\in A}=(x_a)_{a\in A}$.
\end{proof}

\begin{example}
{\itshape  Example of a sharply orthocomplete Archimedean $MV$-effect algebra  
that is not  complete.}%

\medskip

{\rm If we omit in Lemma \ref{soaac} the assumption of atomicity in $E$ 
it is enough to take the  $MV$-effect algebra 
$E=\{f:[0, 1]\to [0, 1] \mid f\ \mbox{continuous}\ \mbox{function}\}$ which is 
a sub-lattice effect algebra of a direct product of copies of the standard 
$MV$-effect algebra of real numbers  $[0, 1]$
that is  Archimedean, 
sharply orthocomplete (the center $C(E)=S(E)=\{0, 1\}$ is 
trivial) and hence sharply dominating. 
Moreover, $E$ is not complete.}
\end{example}

It is well known that an Archimedean  lattice effect algebra $E$
is complete if and only if every block of $E$ 
is complete (see \cite[Theorem 2.7]{ZR60}). If moreover $E$ is atomic 
then $E$ may have atomic as well non-atomic blocks \cite{beltrametti}. 
K. Mosn\'a\ \cite[Theorem 8]{mosna} has proved that in this case 
$E=\bigcup\{M\subseteq E\mid M \ \mbox{atomic\ \mbox{block}\ \mbox{of}\ 
E}\}$. 

Hence every non-atomic block of $E$ is covered by atomic ones. Moreover, 
many properties of  Archimedean  atomic lattice effect algebras as 
well as their non-atomic blocks depend on properties of their 
atomic blocks. 

Namely, the center $C(E)$, the compatibility center $B(E)$ 
and the set $S(E)$ of sharp elements of 
Archimedean  atomic lattice effect algebras $E$ can be expressed 
by set-theoretical operations on  their 
atomic blocks. As follows, 
$B(E)=\bigcap\{M\subseteq E\mid M \ \mbox{atomic\ \mbox{block}\ \mbox{of}\ 
E}\}$,  
$S(E)=\bigcup\{C(M)\mid M\subseteq E, M \ \mbox{atomic\ \mbox{block}\ \mbox{of}\ 
E}\}$ and $C(E)=B(E)\cap S(E)$ (see \cite{mosna}).

For instance, an Archimedean  atomic lattice effect algebra $E$ 
is sharply dominating iff every atomic block of $E$ 
is sharply dominating (see \cite{kopr}).
Moreover, we can prove the following:

\begin{theorem}\label{cahract}
Let $E$ be an Archimedean atomic lattice effect algebra. 
Then the following conditions are 
equivalent: 
\begin{enumerate}\itemsep=0pt
\item[$(i)$] $E$ is complete.
\item[$(ii)$]  Every atomic block of $E$ is complete.
\end{enumerate}
In this case every block of $E$ is complete.
\end{theorem}
\begin{proof}  $(i)\implies (ii)$: This is trivial, as every 
block $M$ of $E$ is a full sub-lattice effect algebra of  $E$.

\noindent{}$(ii) \implies  (i)$: It is enough to show that $E$ 
is orthocomplete.  From \cite[Theorem 2.6]{ZR60} we then get that 
$E$ is complete.

Let $G\subseteq E$ be a $\bigoplus$-orthogonal 
system. Then, for every $x\in G$, there is a set $A_x$ of atoms 
of $E$ and positive integers 
$k_{a}$, $a\in A_x$ such that such that 
$x=\bigoplus_{E} \{k_a a \mid a\in A_x\}$. Moreover, for any 
$F\subseteq G$ finite we have that $\bigcup\{ A_x \mid x\in F\}$ 
is an orthogonal set of atoms. Hence 
$A_G= \bigcup\{ A_x \mid x\in G\}$ is an orthogonal set of atoms of $E$ 
and there is a maximal orthogonal set $A$ of atoms of $E$ such that 
$A_G\subseteq A$. Therefore there is an atomic block $M$ of $E$ 
with $A\subseteq M$. By assumption $\bigoplus_{M} G$ exists and 
$\bigoplus_{M} G=\bigoplus_{E} G$, as $M$  is bifull in $E$ 
because $E$ is Archimedean and atomic (see \cite{PR6}).
\end{proof}

\begin{theorem}\label{rexsoaacea} 
Let $E$ be a sharply orthocomplete lattice effect algebra. 
Then  
\begin{enumerate}\itemsep=0pt
\item[$(i)$] $S(E)$ is a complete orthomodular lattice  bifull  in $E$. 
\item[$(ii)$] $C(E)$ is a complete Boolean algebra  bifull  in $E$. 
\item[$(iii)$] $E$ is sharply dominating, centrally dominating and 
S-dominating.
\item[$(iv)$] If moreover $E$ is Archimedean and atomic then 
$E$ is a complete lattice effect algebra.
\end{enumerate}
\end{theorem}
\begin{proof} $(i), (iii)$:  Let $S\subseteq S(E)$, $S$ orthogonal. Then, for 
any $s\in S$, $s\leq s$. Hence  (since $S(E)$ is full in $E$)
$\bigoplus_{E} S=\bigvee_{E} S=\bigvee_{S(E)} S\in S(E)$ exists. 
Since   $S(E)$ is an Archimedean lattice effect algebra we have 
from \cite[Theorem 2.6]{ZR60} 
that   $S(E)$ is complete.  Moreover, let  $x\in E$ and let 
$G=(w_{\kappa})_{\kappa\in H}$, $w_{\kappa} \in S(E), {\kappa\in H}$ 
be a maximal orthogonal system of mutually different elements such that 
$w_x=\bigoplus_{E} \{w_{\kappa} \mid \kappa\in H\}\leq x$. 
Let us show that $y\in S(E)$, 
$y\leq x$ $\Longrightarrow $ $y\leq w_x\in S(E)$. Clearly, 
$w_x\in S(E)$. Assume that $y\not\leq w_x$. Then 
$w_x < y\vee w_x \leq x$. Hence $z=(y\vee w_x)\ominus w_x \not =0$ 
and $G\cup \{z\}$ is an 
orthogonal system of mutually different elements such that 
$y\vee w_x=w_x\oplus z=\bigoplus_{E} \{w_{\kappa} \mid \kappa\in H\}\oplus z\leq x$, 
a contradiction with the maximality of $G$. Therefore 
$y\leq w_x$ and $E$ is {sharply dominating}, hence 
S-dominating and from Theorem \ref{shoiscc} we get that 
$E$ is centrally dominating. From Theorem \ref{shdisbifull},  
we get that $S(E)$ is bifull  in $E$.

\medskip

\noindent $(ii)$: It follows from $(i), (iii)$ and Theorem \ref{shoiscc}.

\medskip

\noindent $(iv)$: Assume now that $E$ is a sharply orthocomplete 
Archimedean atomic lattice effect algebra. Then every atomic block $M$ 
of $E$ is sharply orthocomplete Archimedean atomic $MV$-effect algebra 
and hence it is a complete  $MV$-effect algebra by Lemma \ref{soaac}. 
By Theorem \ref{cahract}, $E$ is a complete  lattice effect algebra.
\end{proof}

\begin{theorem}\label{soaacea} 
Let $E$ be an atomic lattice effect algebra. 
Then the following conditions are 
equivalent: 
\begin{enumerate}\itemsep=0pt
\item[$(i)$] $E$ is complete.
\item[$(ii)$] $E$ is Archimedean and sharply orthocomplete.
\end{enumerate}
\end{theorem}
\begin{proof}  $(i)\implies (ii)$: By \cite[Theorem 3.3]{ZR54} 
we have that any complete lattice effect algebra is Archimedean. 
Evidently, any complete lattice effect algebra is 
sharply orthocomplete.

\noindent{}$(ii) \implies  (i)$: It follows from Theorem 
\ref{rexsoaacea},  $(iv)$.
\end{proof}

\subsection*{Acknowledgements}
The work of the first author was supported by the 
Slovak Research and Development
Agency under the contract No. 
APVV-0375-06 and by VEGA grant agency, grant number  1/0373/08. 
The second author gratefully
acknowledges financial support 
of the  Ministry of Education of the Czech Republic
under the project MSM0021622409. The third author 
was supported by the Slovak Research and Development
Agency under the contract No. APVV--0071--06.

\pdfbookmark[1]{References}{ref}
\LastPageEnding

\end{document}